\documentclass[conference]{IEEEtran}
\IEEEoverridecommandlockouts
\usepackage{cite}
\usepackage{amsthm}
\usepackage{amsmath,amssymb,amsfonts}
\usepackage{textcomp}
\usepackage{xcolor}
\newtheorem{theorem}{Theorem}[section]

\newtheorem{lemma}[theorem]{Lemma}
\newtheorem{corollary}[theorem]{Corollary}

\newtheorem{definition}[theorem]{Definition}
\newtheorem{remark}{Remark}
\newtheorem{example}[theorem]{Example}

\def\BibTeX{{\rm B\kern-.05em{\sc i\kern-.025em b}\kern-.08em
    T\kern-.1667em\lower.7ex\hbox{E}\kern-.125emX}}
\begin{document}

\title{New non-binary quantum codes from skew constacyclic codes over the ring $\mathbb{F}_{p^m}+v\mathbb{F}_{p^m}+v^2 \mathbb{F}_{p^m}$\\
}

\author{\IEEEauthorblockN{ Ram Krishna Verma\textsuperscript{st}}
\IEEEauthorblockA{\textit{Department of Mathematics} \\
\textit{Indian Institute of Technology Patna}\\
Patna- 801106, India\\
email-ram.pma15@iitp.ac.in}
\and
\IEEEauthorblockN{ Om Prakash\textsuperscript{nd}}
\IEEEauthorblockA{\textit{Department of Mathematics} \\
	\textit{Indian Institute of Technology Patna}\\
	Patna- 801106, India\\
	email-om@iitp.ac.in}
\and
\IEEEauthorblockN{ Ashutosh Singh\textsuperscript{rd}}
\IEEEauthorblockA{\textit{Department of Mathematics} \\
	\textit{Indian Institute of Technology Patna}\\
	Patna- 801106, India\\
	email-ashutosh\_1921ma05@iitp.ac.in}

}

\maketitle

\begin{abstract} In this article, we construct new non-binary quantum codes from skew constacyclic codes over finite commutative non-chain ring $\mathcal{R}= \mathbb{F}_{p^m}[v]/\langle v^3 =v \rangle$ where $p$ is an odd prime and $m \geq 1$.
	In order to obtain such quantum codes, first we study the structural properties of skew constacyclic codes and their Euclidean duals over the ring $\mathcal{R}$. Then a necessary and sufficient condition for skew constacyclic codes over $\mathcal{R}$ to contain their Euclidean duals is established. Finally, with the help of CSS construction and using Gray map, many new non-binary quantum codes are obtained over $\mathbb{F}_{p^m}$.
\end{abstract}


\section{Introduction}
Quantum error-correcting codes have been of interest to many researchers and extensively studied due to their applications in quantum communication. These codes are advantageous to protect the information carried by quantum states on the quantum channels against environmental and operational noise(decoherence). In $1995$, Shore \cite{Shor} introduced first quantum code with parameters $[[9,1,3]]$. In $1998$, Calderbank et al. \cite{Calderbank} discussed a method for construction of binary quantum codes from the classical self orthogonal codes over $GF(4)$. Further, in $1999$, E.M. Rains \cite{Rains} generalized the concept for the non-binary codes and constructed non-binary stabilizer quantum codes from classical linear codes over $\mathbb{F}_q$. Thereafter, several optimal quantum codes have been constructed by using classical linear codes over finite field (see\cite{Gaurdia10,Gottesman10,Grassl04,Li04,Li10}).\\

It has also noticed that classical linear (cyclic or constacyclic) codes over the finite rings can be viewed as an excellent resource to produce many good quantum codes. Recently, many works have been reported on codes over finite rings in order to construct quantum codes. In $2009$, Qian et al. \cite{Qian09} obtained many binary quantum codes from cyclic codes over the ring $\mathbb{F}_2+u\mathbb{F}_2$, where $u^2=0$.  Further, Kai and Zhu \cite{Kai11} presented several new quantum codes from cyclic codes over the ring $\mathbb{F}_4+u\mathbb{F}_4$ with $u^2=0$.

In 2015, Ashraf and Mohammad \cite{Ashraf1} constructed quantum codes over $\mathbb{F}_{p}$ from cyclic codes over the non-chain ring $\mathbb{F}_{p}+v\mathbb{F}_{p}$. Dertli et al. \cite{Dertli15} studied cyclic codes over the ring $\mathbb{F}_{2}+u\mathbb{F}_{2}+v\mathbb{F}_{2}+uv\mathbb{F}_{2}$ and constructed some new binary quantum codes. Then, Ashraf and Mohammad \cite{Ashraf2} again generalized their work over the ring $\mathbb{F}_{q}+u\mathbb{F}_{q}+v\mathbb{F}_{q}+uv\mathbb{F}_{q}$ and obtained many new non-binary quantum codes.\\
Constacyclic codes are a generalization of the cyclic codes and often used by many researchers to produce quantum code with good parameters. Gao and Wang \cite{Gao1} discussed the structural properties of $u$-constacyclic codes over the ring $\mathbb{F}_p+u\mathbb{F}_p$ where $u^2=1$ and constructed several new non-binary quantum codes. Later, Ma et al. \cite{Ma} considered the ring $\mathbb{F}_q+v\mathbb{F}_q+v^2\mathbb{F}_q$ and obtained some new quantum codes.\\
On the other side, Boucher et al. \cite{Boucher} generalized cyclic codes as skew cyclic ($\theta$-cyclic) over $\mathbb{F}_{p^m}$, where $\theta$ is the automorphism of $\mathbb{F}_{p^m}$. They considered skew cyclic codes of length $n$ over $\mathbb{F}_{p^m}$ as left ideals of the non-commutative residue polynomial ring $\frac{\mathbb{F}_{p^m}[x; \theta]}{\langle x^n-1 \rangle}$. The vital inspiration for considering codes in this setup is that skew polynomial ring is a not a unique factorization domain. Thus polynomials exhibit more factorizations, and consequently, more ideals than commutative setup. Afterward, many researchers started to study skew cyclic codes over different rings in order to construct linear code with better parameters. In $2015$, Shi et al. \cite{sole} presented structure of skew cyclic codes over the finite commutative non-chain ring $\mathbb{F}_q+v\mathbb{F}_q+v^2\mathbb{F}_q$, where $v^3=v$.\\
Recently, Bag et al. \cite{bag} obtained some non-binary quantum codes from skew constacyclic codes over a non-chain ring $\mathbb{F}_q+u\mathbb{F}_q+v\mathbb{F}_q+uv \mathbb{F}_q $ with $u^2=1, v^2=1$ and $uv=vu$.\\
Motivated by the above studies, here we studied the structure of skew constacyclic codes and their duals over the non-chain ring $\mathcal{R} =\mathbb{F}_q+v\mathbb{F}_q+v^2\mathbb{F}_q$, where $v^3=v$. Further, we used the Gray map defined in \cite{Ma} to get the Gray images of such codes.
%We also obtain a necessary and sufficient condition for skew constacyclic codes over $\mathcal{R}$ to contains their duals. Finally, with the help of Gray map and using CSS construction many new quantum codes have been obtained.\\
The article is organized as follows. Section $2$ discusses the structure of ring $\mathcal{R}$ and review of some results on linear codes over $\mathcal{R}$. Section $3$ studies the skew constacyclic codes over the ring $R$ and their Gray images under the Gray map defined in \cite{Ma}.  A necessary and sufficient condition for skew constacyclic codes over $\mathcal{R}$ to contain their duals has been established in Section $4$. As an application, we get several new non-binary quantum codes included in Table \ref{tab1}. Section $5$ concludes the work.

\section{Preliminaries}
For an odd prime $p$ and positive integer $m$, let $\mathbb{F}_{p^m}$ be the finite field with characteristic $p$. Throughout the article, $\mathcal{R}$ denotes the ring $\mathbb{F}_{p^m}+v\mathbb{F}_{p^m}+v^2 \mathbb{F}_{p^m}$ where $v^3=v$. Clearly, the ring $\mathcal{R}$ is a finite commutative non-chain ring with characteristic $p$ and order $p^{3m}$. Moreover, the ring $\mathcal{R}$ is semi-local principal ideal ring with three maximal ideals $\langle v \rangle, \langle v+1 \rangle$ and $\langle v-1 \rangle$.\\
 Let $2 \zeta \cong 1 (mod\ p)$
and  $\eta_0=1-v^2, \eta_1= \zeta (v^2+v), \eta_2=\zeta (v^2-v)$. Then $\sum_{i=0}^{2} \eta_i=1$ in $\mathcal{R}$ and
\begin{equation*}
\eta_i \eta_j=\begin{cases}
\eta_i, & \text{if $i=j$}\\
0, & \text{if $i\neq j$}
\end{cases}.
\end{equation*}
Therefore, by Chinese Remainder Theorem, we decompose the ring $\mathcal{R}$ as
\begin{align*}
\mathcal{R}=&~\eta_0\mathcal{R}+ \eta_1\mathcal{R}+ \eta_2\mathcal{R}\\
 \cong &~ \eta_0 \mathbb{F}_{p^m}+ \eta_1\mathbb{F}_{p^m}+ \eta_2\mathbb{F}_{p^m}
\end{align*}
Hence, any element $r=a+bv+cv^2 \in \mathcal{R}$ can be uniquely expressed as $r=\beta_0 \eta_0+\beta_1  \eta_1+\beta_2  \eta_2$ where $\beta_0,\beta_1,\beta_2 \in \mathbb{F}_{p^m}$.\\
Let $M\in GL_{3}(\mathbb{F}_{p^m})$, where $GL_{3}(\mathbb{F}_{p^m})$ is the set of all $3\times 3$ invertible matrices over $\mathbb{F}_{p^m}$. Now, by the Gray map defined in \cite{Ma},
\begin{align*}
\psi: \mathcal{R} \longrightarrow \mathbb{F}_{p^m}^{3}
\end{align*}
by
\begin{align}\label{map 1}
r\longmapsto&(\beta_{0},\beta_{1},\beta_{2})M\\
=& \boldsymbol{r}M
\end{align}
Here, we use $\boldsymbol{r}$ for the vector $(\beta_{0},\beta_{1},\beta_{2})$. It is easy to check that the map $\psi$ is linear and can be naturally extended from $\mathcal{R}^{n}$ to $\mathbb{F}_{p^m}^{3 n}$ componentwise.

Recall that nonempty subset $C$ of $\mathcal{R}^{n}$ is called a linear code of length $n$ over $\mathcal{R}$ if it is an $\mathcal{R}$-submodule of $\mathcal{R}^{n}$ and elements of $C$ are called codewords. The dual of a linear code ${C}$ of length $n$ over $\mathcal{R}$ is defined as ${C}^{\perp}=\{a \in \mathcal{R}^{n}\mid a \cdot b=0~\forall~b\in {C} \}$ is also a linear code. The code ${C}$ is said to be self-orthogonal if $ {C} \subseteq {C}^{\perp} $ and self-dual if ${C}^{\perp}= {C}$. In order to consider dual code, we use Euclidean inner product. For any codeword $c=(c_{0},c_{1},\dots,c_{n-1})\in \mathcal{C}$, the Hamming weight $w_{H}(c)$ is define as the number of non-zero components in $c$ and the Hamming distance between two codewords $c$ and $c'$ is given by $d_{H}(c,c')=w_{H}(c-c')$. The Hamming distance for a code $\mathcal{C}$ is define as $d_{H}(\mathcal{C})=min\{d_{H}(c,c')\mid c\neq c', \forall~c,c'\in \mathcal{C} \}$. It is easy to check that the Gray map $\psi$ is distance preserving linear map.\\
Now, we review some results on linear codes over $\mathcal{R}$ which can be easily proved by using basic concept of algebraic coding theory.

Let $\mathcal{C}$ be a linear code of length $n$ over $\mathcal{R}$. Define
\begin{align*}
\mathcal{A}_{0}=\{x_{0}\in \mathbb{F}_{p^m}^{n}\mid \exists x_{1},x_{2} \in  \mathbb{F}_{p^m}^{n} , \eta_{0}x_{0}+\eta_{1}x_{1}+\eta_{2}x_{2} \in {C}\};\\
\mathcal{A}_{1}=\{x_{1}\in \mathbb{F}_{p^m}^{n}\mid \exists x_{0},x_{2} \in  \mathbb{F}_{p^m}^{n},\eta_{0}x_{0}+\eta_{1}x_{1}+\eta_{2}x_{2} \in {C}\};\\
\mathcal{A}_{2}=\{x_{2}\in \mathbb{F}_{p^m}^{n}\mid \exists x_{0},x_{1} \in  \mathbb{F}_{p^m}^{n},\eta_{0}x_{0}+\eta_{1}x_{1}+\eta_{2}x_{2} \in {C}\}.
\end{align*}
Then each $\mathcal{A}_{i}$ is a linear code over $ \mathbb{F}_{p^m}$ for $i=0,1,2$ and ${C}= \eta_{0}\mathcal{A}_{1}\oplus \eta_{1}\mathcal{A}_{1} \oplus \eta_{2}\mathcal{A}_{2}$. Moreover, if $G_i$ is the generator matrix for a linear code $\mathcal{A}_{i}$ over $ \mathbb{F}_{p^m}$, then the generator matrix $G$ for linear code $\mathcal{C}$ over $\mathcal{R}$ can be written as  $$G=\begin{pmatrix}
\eta_{0}G_{0} \\
\eta_{1}G_{1} \\
\eta_{2}G_{2} \\
\end{pmatrix}$$.

\begin{lemma}\label{cor k} \cite{gao}
If $\mathcal{C}= \eta_{0}\mathcal{A}_{0}\oplus \eta_{1}\mathcal{A}_{1} \oplus \eta_{2}\mathcal{A}_{2}$ is a linear code of length $n$ over $\mathcal{R}$ where $\mathcal{A}_{i}$ is the $[n, k_{i}, d(\mathcal{A}_{i})]$ linear code over $\mathbb{F}_{p^m}$, then $\phi(\mathcal{C})$ is the $[3n,\sum_{i=0}^{2}k_{i}, d]$ linear code.
\end{lemma}
\begin{lemma} (\cite{jgao}, Proposition $1$) If $\mathcal{C}= \eta_{0}\mathcal{A}_{0}\oplus \eta_{1}\mathcal{A}_{1} \oplus \eta_{2}\mathcal{A}_{2}$ is a linear code of length $n$ over $\mathcal{R}$, then $\mathcal{C}^{\perp}= \eta_{0}\mathcal{A}_{0}^{\perp}\oplus \eta_{1}\mathcal{A}_{1}^{\perp} \oplus \eta_{2}\mathcal{A}_{2}^{\perp}$. Moreover, $C$ is a self-orthogonal code over $\mathcal{R}$ if and only if each $\mathcal{A}_{i}$ is self-orthogonal code over $\mathbb{F}_{p^m}$ for $i=0,1,2$.
	
\end{lemma}
Next result illustrates units in the ring $\mathcal{R}$ can be found in Theorem $1$ of \cite{Ma}.

 \begin{lemma}
 	Let $\delta=\alpha+v\beta+v^2 \gamma \in \mathcal{R}$. Then $\delta$ is unit in $\mathcal{R}$ if and only if $\alpha$, $\alpha+\beta+\gamma$ and $\alpha-\beta+\gamma$ are units in $ \mathbb{F}_{p^m}$.
 \end{lemma}
Throughout the article, $\delta=\alpha+v\beta+v^2 \gamma$ represents a unit in $\mathcal{R}$ and $\lambda_0=\alpha$, $\lambda_1=\alpha+\beta+\gamma$ and $\lambda_2=\alpha-\beta+\gamma$ represent corresponding units in $ \mathbb{F}_{p^m}$.

\section{Structure of skew constacyclic codes over the ring $\mathcal{R}$}
We begin the section with structure of skew polynomial ring over $\mathcal{R}$ and some definition and results for the skew polynomial rings over $\mathbb{F}_{p^m}$ which are useful to study the structure of skew constacyclic codes over the ring $\mathcal{R}$.

 Let $Aut(\mathbb{F}_{p^m})$ be the set of all automorphisms over
  the field $\mathbb{F}_{p^m}$ and $\Theta \in Aut(\mathbb{F}_{p^m})$ be a non-trivial automorphism. Now, define a map $\sigma : \mathcal{R} \longrightarrow \mathcal{R}$ as
 $$\sigma(a+bv+cv^2)= \Theta(a)+\Theta(b)v+\Theta(c)v^2,$$
 where $a,b,c \in \mathbb{F}_{p^m}$. Clearly $\sigma$ is an automorphism of $\mathcal{R}$. Now, we define the skew polynomial ring $$\mathcal{R}[x;\sigma]:=\{a_{0}+a_{1}x+\cdots+a_{n}x^{n}\mid a_{i}\in \mathcal{R}~ \forall~ i, n \in  \mathbb{N} \}$$ equipped with usual addition of polynomials and multiplication of polynomials is defined under the rule $(ax^{i})(bx^{j})=a\sigma^{i}(b)x^{i+j}$. Clearly, $\mathcal{R}[x;\sigma]$ is a noncommutative ring. If $\Theta$ is the identity automorphism, then  $\mathcal{R}[x;\sigma]$ is simply the polynomial ring $\mathcal{R}[x]$.
\begin{definition}
	Let $\delta$ be a unit in $\mathcal{R}$ and $\sigma\in Aut(\mathcal{R})$. Suppose $\tau_{\sigma,\delta}$ be a skew $\delta$-constacyclic shift operator from $ \mathcal{R}^n \longrightarrow \mathcal{R}^n$ defined by $\tau_{\sigma,\delta}(c)=(\sigma(\delta c_{n-1}),\sigma(c_{0}),\dots,\sigma(c_{n-2}))\in C$ for $c=(c_{0},c_{1},\dots,c_{n-1})\in C$.\\
	Then a linear code $C$ of length $n$ over $\mathcal{R}$ is called skew $(\sigma,\delta)$-constacyclic code if $\tau_{\sigma,\delta}(C)=C$
	In particular, for $\delta=1$ and $\delta=-1$, $C$ is called skew cyclic and skew negacyclic code, respectively. In case, if $\sigma$ is the identity automorphism, then $C$ is a constacyclic code.
\end{definition}{}
 The center of skew polynomial ring $\mathcal{R}[x;\sigma]$ is denoted by $Z(\mathcal{R}[x;\sigma])$.
\begin{theorem}\cite{sole}\label{qw}
Let $g(x)=x^n-1 \in \mathcal{R}[x;\sigma]$. Then $g(x) \in Z( \mathcal{R}[x;\sigma])$ if and only if the order of automorphism $\sigma$ divides $n$. Moreover, for a unit $\delta$ in $\mathcal{R}$, the polynomial $x^n-\delta \in Z( \mathcal{R}[x;\sigma])$ if and only if $\sigma$ fixes $\delta$ and order of $\sigma$ divides $n$.
\end{theorem}
If $C$ is a skew $(\sigma,\delta)$-constacyclic code of length $n$ over $\mathcal{R}$, then similar to the constacyclic codes, we can identify each codeword $c=(c_{0},c_{1},\dots,c_{n-1})\in C$ by a polynomial $c(x)=c_{0}+c_{1}x+\dots+c_{n-1}x^{n-1}$ in the quotient $\mathcal{R}[x;\sigma]/\langle x^{n}-\delta\rangle$ under the correspondence $c=(c_{0},c_{1},\dots,c_{n-1})\mapsto c(x)=(c_{0}+c_{1}x+\dots+c_{n-1}x^{n-1})$ $mod$ $(x^{n}-\delta)$. From Theorem \ref{qw}, if $\delta $ fixed by $\sigma$ and order of $\sigma$ divides $n$, then $\langle x^{n}-\delta\rangle$ forms two sided ideal in $\mathcal{R}[x;\sigma]$. Therefore, the structure $\mathcal{R}[x;\sigma]/\langle x^{n}-\delta\rangle$ forms residue polynomial ring. Thus, a skew $(\sigma,\delta)$-constacyclic code of length $n$ over $\mathcal{R}$ can be defined as a left ideal in the quotient ring $\mathcal{R}[x;\sigma]/\langle x^{n}-\delta\rangle$ generated by right divisor of $( x^{n}-\delta)$. However, in the case when order of $\sigma$ does not divide $n$ then the quotient $\mathcal{R}[x;\sigma]/\langle x^{n}-\delta\rangle$ is not a ring. In this case, it forms left $\mathcal{R}[x;\sigma]$ module, where multiplication is defined by $$r(x)*(l(x)-(x^{n}-\delta))=r(x)*l(x)+( x^{n}-\delta ).$$ Therefore, a skew $(\sigma,\delta)$-constacyclic code of length $n$ over $\mathcal{R}$ can be considered as $\mathcal{R}[x;\sigma]$-submodule of module $\mathcal{R}[x;\sigma]/\langle x^{n}-\delta\rangle$. In both cases, $\mathcal{C}$ is generated by monic right divisors of $x^n-\delta$ in $\mathcal{R}[x;\sigma]$.
	Hence, for further, we assume that order of automorphism $\sigma$ divides $n$ and $\sigma$ fixes $\delta$.

\begin{theorem}\cite{Ma}\label{thm field}
	Let $\lambda \in \mathbb{F}^{*}_{p^m}$ and $\Theta \in Aut(\mathbb{F}_{p^m})$. Let $C$ be a linear code of length $n$ over $\mathbb{F}_{p^m}$. Then $C$ is skew $(\Theta,\lambda)$-constacyclic over  $\mathbb{F}_{p^m}$ if and only if there exists a polynomial $g(x) \in \mathbb{F}_{p^m}[x;\Theta]/\langle x^{n}-\lambda\rangle$ such that $C= \langle g(x)\rangle$ and $g(x)$ is a right divisor of $(x^{n}-\lambda)$
	in $\mathbb{F}_{p^m}[x;\Theta]$.
\end{theorem}{}

\begin{theorem}
	Let $C= \sum_{i=0}^{2} \eta_i \mathcal{A}_i$ be a linear code of length $n$ over $\mathcal{R}$. Then $C$ is a skew $(\sigma,\delta)$-constacyclic code if and only if $\mathcal{A}_i$ is skew $(\Theta,\lambda_i)$-constacyclic code over $\mathbb{F}_{p^m}$ for $i=0,1,2$.
\end{theorem}
\begin{proof}
	Let $(a_{i,0},a_{i,1}, \cdots, a_{i,{n-1}} ) \in \mathcal{A}_i$ and $a_j= \sum_{i=1}^{2} \eta_i a_{i,j}$ for $j=0,1, \cdots, {n-1}$. Then $a=(a_0,a_1, \cdots, a_{n-1}) \in C$. Suppose $C$ is skew $(\sigma,\delta)$-constacyclic code of length $n$ over $\mathcal{R}$, then for any codeword $c \in C$, we have $\tau_{\sigma,\delta}(c) \in C$.
	Now,
	\begin{align*}
	\tau_{\sigma,\delta}(a)=&~(\delta \sigma(a_{n-1}),\sigma(a_0), \cdots, \sigma(a_{n-2}))\\
	=&~ \eta_0(\lambda_0 \Theta(a_{0,n-1}),\Theta(a_{0,0}), \cdots, \Theta(a_{0,n-2}))\\
	+&~ \eta_1(\lambda_1 \Theta(a_{1,n-1}),\Theta(a_{1,0}), \cdots, \Theta(a_{1,n-2}))\\
	+&~ \eta_2(\lambda_2 \Theta(a_{2,n-1}),\Theta(a_{2,0}), \cdots, \Theta(a_{2,n-2}))\\
	\in&~ C= \sum_{i=0}^{2} \eta_i \mathcal{A}_i.
	\end{align*}
	As the decomposition of linear codes over the ring $\mathcal{R}$ is unique. Therefore,
	$(\lambda_i \Theta(a_{i,n-1}),\Theta(a_{i,0}), \cdots, \Theta(a_{i,n-2})) \in \mathcal{A}_i$ for $i=0,1,2$. Consequently, $\mathcal{A}_i$ is skew $(\Theta,\lambda_i)$-constacyclic code over $\mathbb{F}_{p^m}$ for $i=0,1,2$.
	Conversely, suppose $\mathcal{A}_i$ is skew $(\Theta,\lambda_i)$-constacyclic code of length $n$ over $\mathbb{F}_{p^m}$ for $i=0,1,2$, and $a=(a_0,a_1, \cdots, a_{n-1}) \in C$ where $a_j= \sum_{i=1}^{2} \eta_i a_{i,j}$ for $j=0,1, \cdots, {n-1}$. Then $(a_{i,0},a_{i,1}, \cdots, a_{i,{n-1}} ) \in \mathcal{A}_i$. Therefore, from above equation, $\tau_{\sigma,\delta}(a) \in C$, i.e., $C$ is skew $(\sigma,\delta)$-constacyclic code over $\mathcal{R}$.
\end{proof}

In the next result we will provide generator polynomial of skew $(\sigma,\delta)$-constacyclic code over $\mathcal{R}$ in terms of generator polynomial of skew $(\Theta,\lambda_i)$-constacyclic code over $\mathbb{F}_{p^m}$ for $i=0,1,2$.

\begin{theorem}
	Let $C= \sum_{i=0}^{2} \eta_i \mathcal{A}_i$ be a skew $(\sigma,\delta)$-constacyclic code of length $n$ over $\mathcal{R}$ and $f_i(x)$ is the generator polynomial of skew $(\Theta,\lambda_i)$-constacyclic code $\mathcal{A}_i$ over $\mathbb{F}_{p^m}$ for $i=0,1,2$, respectively. Then
	\begin{itemize}
		\item[1.] there exist a polynomial $f(x) \in \mathcal{R}[x;\sigma]$ such that $C =\langle f(x) \rangle$ and $(x^n-\delta)$ is right divisible by $f(x)$, where $f(x)=\sum_{i=0}^{2} \eta_i f_i(x)$.
		\item[2.] $C= \langle \eta_0 f_0(x), \eta_1 f_1(x) , \eta_2 f_2(x)  \rangle $ and $\mid C \mid = p^{{3mn}- \sum_{i=0}^{2} \text{deg}f_i}$.
	\end{itemize}
\end{theorem}
\begin{proof}
	\begin{itemize}
	\item[1.] Since $C= \sum_{i=0}^{2} \eta_i \mathcal{A}_i$ is a skew $(\sigma,\delta)$-constacyclic code of length $n$ over $\mathcal{R}$ and $f_i(x)$ is the generator polynomial of skew $(\Theta,\lambda_i)$-constacyclic code $\mathcal{A}_i$ over $\mathbb{F}_{p^m}$ for $i=0,1,2$, respectively. Therefore, $\eta_i f_i(x) \in \eta_i \mathcal{A}_i \subseteq C$, which implies that $\langle \sum_{i=0}^{2}\eta_i f_i(x) \rangle \subseteq C $.
	On the other side, let $g(x)\in C$. As $C= \sum_{i=0}^{2} \eta_i \mathcal{A}_i$, there exist some polynomials $h_i(x) \in \mathbb{F}_{p^m}[x;\Theta]$ for $i=0,1,2$ such that $f(x)=\sum_{i=0}^{2} \eta_i h_i(x) f_i(x)$.  Therefore, $f(x) \in \langle \sum_{i=0}^{2}\eta_i f_i(x) \rangle $, i.e. $C=\langle \sum_{i=0}^{2}\eta_i f_i(x) \rangle$.\\
	Since, by Theorem \ref{thm field}, $f_i(x)$ right divides $(x^n-\lambda_i)$ for $i=0,1,2$. Therefore, there exist polynomials $h_i(x) \in \mathbb{F}_{p^m}[x;\Theta]$ such that $(x^n-\lambda_i) =h_i(x) f_i(x)$ for $i=0,1,2$. Now, $(\sum_{i=0}^{2}\eta_i f_i(x)) (\sum_{i=0}^{2}\eta_i h_i(x)) = (x^n-\delta)$. Thus, $\sum_{i=0}^{2}\eta_i f_i(x)=f(x)$ is a right divisor of $(x^n-\delta)$.
	
	\item[2.] From first part, we can write $C= \langle \eta_0 f_0(x), \eta_1 f_1(x) , \eta_2 f_2(x)  \rangle $. Since $C= \sum_{i=0}^{2} \eta_i \mathcal{A}_i$. Thus
	\begin{align*}
	\mid C\mid =&\mid \mathcal{A}_0 \mid \mid \mathcal{A}_1 \mid \mid \mathcal{A}_2 \mid\\
	=&~ p^{mn- \text{deg}f_0(x)}p^{mn -\text{deg}f_1(x)} p^{mn- \text{deg}f_2(x)}\\
	=&~ p^{{3mn}- \sum_{i=0}^{2} \text{deg}f_i(x)}.
	\end{align*}
\end{itemize}
\end{proof}
\begin{remark}
	For a polynomial $h(x)= \sum\limits_{j=0}^k h_{j} x^j \in \mathcal{R}$ of degree $k$, the skew reciprocal polynomial of $h(x)$ is defined as  $h^*(x):= \sum\limits_{j=0}^k \sigma^j(h_{k-j}) x^j$.
\end{remark}
In next two results we discuss the structure of Euclidean dual of skew $(\sigma,\delta)$-constacyclic code of length $n$ over $\mathcal{R}$.

\begin{theorem}
If $C= \sum_{i=0}^{2} \eta_i \mathcal{A}_i$ is a skew $(\sigma,\delta)$-constacyclic code of length $n$ over $\mathcal{R}$, then $C^{\perp}=\sum_{i=0}^{2} \eta_i \mathcal{A}_i^{\perp}$ is skew $(\sigma,\delta^{-1})$-constacyclic code over $\mathcal{R}$ where $\mathcal{A}_i^{\perp}$ is skew $(\Theta,\lambda_i^{-1})$-constacyclic code over $\mathbb{F}_{p^m}$ for $i=0,1,2$.
\end{theorem}

\begin{proof}
 Let $C= \sum_{i=0}^{2} \eta_i \mathcal{A}_i$ be a skew $(\sigma,\delta)$-constacyclic code of length $n$ over $\mathcal{R}$. Since $\delta$ is invariant under $\sigma$ and order of $\sigma$ is a factor of $n$, therefore, by Lemma $(3.1)$ of \cite{Jitman}, $C^{\perp}$ is a skew $(\sigma,\delta^{-1})$-constacyclic code over $\mathcal{R}$.
\end{proof}
\begin{corollary}
Let $C= \sum_{i=0}^{2} \eta_i \mathcal{A}_i$ be a skew $(\sigma,\delta)$-constacyclic code of length $n$ over $\mathcal{R}$ and $f_i(x)$ be the generator polynomial of skew $(\Theta,\lambda_i)$-constacyclic code $\mathcal{A}_i$ over $\mathbb{F}_{p^m}$ for $i=0,1,2$, respectively. Then there exists a polynomial $\ell(x) \in \mathcal{R}[x;\sigma]$ such that $C^{\perp}= \langle \ell(x) \rangle$ where $\ell(x)= \sum_{i=0}^{2} \eta_i h_i^{*}(x)$ and $h_i^{*}(x)$ is skew reciprocal polynomial of $h_i(x)$ where $f_i(x) h_i(x)= (x^n-\lambda_i)$ for $i=0,1,2$.
\end{corollary}

\subsection{Gray image of skew constacyclic codes over the ring $\mathcal{R}$}
The aim of this subsection is to demonstrate the Gray image of skew $(\sigma,\delta)$-constacyclic code over $\mathcal{R}$. Towards this, we begin with the following definition.
\begin{definition}
	Let $C$ be a linear code of length $n=st$ (where $s,t$ are positive integer) over $\mathcal{R}$ and $\delta$ is unit in $\mathcal{R}$. Let $\Upsilon_{\sigma,t}: \mathcal{R}^n \longrightarrow \mathcal{R}^n$ be a linear operator defined by
	\begin{align*}
	\Upsilon_{\sigma,t}(a)=&~(a^1 \mid a^2 \mid \cdots\mid a^t)\\
	 =&~ (	\tau_{\sigma,\delta}(a^1)  \mid \tau_{\sigma,\delta}(a^2) \mid \cdots \mid \tau_{\sigma,\delta}(a^t))
 \end{align*}
 where $a^i \in \mathcal{R}^s$ for $i=1,2, \cdots, t$. Then $C$ is called a skew quasi twisted code of length $n$ and index $t$ if $\Upsilon_{\sigma,t}(C)=C$. If $\sigma$ is the identity automorphism, then $C$ is a quasi twisted code of length $n$ and index $t$ over $\mathcal{R}$.
\end{definition}
Next theorem easily follows from definition of skew quasi twisted code.
\begin{theorem}
	Let $C$ be a skew $(\sigma,\delta)$-constacyclic code of length $n$ over $\mathcal{R}$. Then $\psi(C)$ is a skew quasi twisted code of length $3n$ and index $3$ over $\mathbb{F}_{p^m}$.
\end{theorem}

\section{Quantum codes from skew constacyclic codes over the ring $\mathcal{R}$}

In this section, we construct several new non-binary quantum codes over finite field $\mathbb{F}_{p^m}$ with the help of dual containing skew constacyclic codes over the ring $\mathcal{R}$. For a prime $p$ and positive integer $m$, a $p^m$-ary quantum code $Q$ of length $n$ is a $p^{mk}$ dimensional subspace of $p^{mn}$ dimensional complex Hilbert space $(\mathbb{C}^{p^m})^{\otimes n}=\underbrace{\mathbb{C}^{p^m} \otimes \mathbb{C}^{p^m} \dots \otimes \mathbb{C}^{p^m}}_{n-times} $ and rigorously represented by $[[n,k,d]]_{p^m}$ where $d$ is the minimum distance of $Q$. It can correct both types of errors i.e., bit flip and phase shift errors up to $\lfloor \frac{d-1}{2}\rfloor$.\\
In $1996$, Calderbank et al. \cite{Calderbank} gave a method for the construction of binary quantum codes from classical linear codes. Further, E.M. Rains \cite{Rains} and Ketkar et al. \cite{Ketkar} generalized for non binary case and constructed stabilizer quantum codes from classical linear codes over $\mathbb{F}_{p^m}$.\\
Now, we recall the well known result known as CSS construction (Lemma \ref{lemma css}) which plays a vital role in the construction of quantum codes.

\begin{lemma}[\cite{Grassl04}, Theorem 3] \label{lemma css}
	If $C_{1}=[n,k_{1},d_{1}]_{p^m}$ and $C_{2}=[n,k_{2},d_{2}]_{p^m}$ are two linear codes over $GF(p^m)$ such that $C_{2}^{\perp}\subseteq C_{1}$, then there exists a QECC with parameters $[[n,k_{1}+k_{2}-n,d]]$ where $d=min\{w(v): v\in (C_{1}\backslash C_{2}^{\perp})\cup (C_{2}\backslash C_{1}^{\perp})\}\geq min\{ d_{1},d_{2}\}$. Moreover, if $C_{1}^{\perp}\subseteq C_{1},$ then there exists a quantum code $C$ with parameters $[[n,2k_{1}-n,d_{1}]]$, where $d_{1}=min\{w(v):v\in C_{1} \backslash C_{1}^{\perp}\}$.
\end{lemma}{}

To construct quantum codes over $\mathbb{F}_{p^m}$ from skew constacyclic codes over the ring $\mathcal{R}$, we obtain a necessary and sufficient condition for skew constacyclic codes over the ring $\mathcal{R}$ to contain their duals. The following result follows the similar argument of Theorem 5.4 and Theorem 5.5 of \cite{bag}.

	\begin{lemma}\label{thm dual} Let $\mathcal{A}_i$ be a skew $(\Theta, \lambda_i)$-constacyclic code of length $n$ over $\mathbb{F}_{p^m}$ with generator polynomials $f_i(x)$ for $i=0,1,2$. Then $\mathcal{A}_i$ contains its dual if and only if $h_i^{*}(x) h_i(x)$ is right divisible by $(x^n-\lambda_i)$, where $(x^n-\lambda_i)=h_i(x) f_i(x)$, $h_i^{*}(x)$ is the skew reciprocal polynomial of $h_i(x)$ and $\lambda_i =\pm 1$ for $i=0,1,2$.
\end{lemma}

\begin{theorem}
	Let $C= \sum_{i=0}^{2} \eta_i \mathcal{A}_i$  be a skew $(\sigma,\delta)$-constacyclic code of length $n$ over $\mathcal{R}$ with generator polynomial $f(x)=\sum_{i=0}^{2} \eta_i f_i(x)$, where $f_i(x)$ is the generator polynomial of skew $(\Theta,\lambda_i)$-constacyclic code $ \mathcal{A}_i$ over $\mathbb{F}_{p^m}$ for $i=0,1,2$ with $\lambda_{i}=\pm 1$. Then $C^{\perp} \subseteq C$ if and only if $(x^n-\lambda_i)$ right divides $h_i^{*}(x) h_i(x)$ for all $i=0,1,2$.
	Here, $(x^n-\lambda_i)=h_i(x) f_i(x)$ and  $h_i^{*}(x)$ represents skew reciprocal of $h_{i}(x)$ for $i=0,1,2$.
\end{theorem}

\begin{proof}
	Let $C= \sum_{i=0}^{2} \eta_i \mathcal{A}_i$  be a skew $(\sigma,\delta)$-constacyclic code of length $n$ over  $\mathcal{R}$ and ${C}^{\perp} \subseteq {C}$. Then $\sum_{i=0}^{2} \eta_i \mathcal{A}_i^{\perp}\subseteq \sum_{i=0}^{2} \eta_i \mathcal{A}_i$. Since $\eta_i$ is primitive orthogonal idempotent in $\mathcal{R}$, thus by taking modulo $\eta_i$ we get $\mathcal{A}_i^{\perp}\subseteq \mathcal{A}_i$ for $i=0,1,2$. Therefore, by Lemma \ref{thm dual}  $(x^n-\lambda_i)$ right divides $h_i^{*}(x) h_i(x)$ for $i=0,1,2$.\\
	Conversely, let $h_i^{*}(x) h_i(x)$ be right divisible by $(x^n-\lambda_i)$ for $i=0,1,2$. Then, by Lemma \ref{thm dual}, we get $\mathcal{A}_i^{\perp}\subseteq \mathcal{A}_i$ for $i=0,1,2$. Therefore, $\sum_{i=0}^{2} \eta_i \mathcal{A}_i^{\perp}\subseteq \sum_{i=0}^{2} \eta_i \mathcal{A}_i$. Consequently, ${C}^{\perp} \subseteq {C}$.
	\end{proof}
\begin{corollary}
Let $C= \sum_{i=0}^{2} \eta_i \mathcal{A}_i$  be a skew $(\sigma,\delta)$-constacyclic code of length $n$ over $\mathcal{R}$. Then ${C}^{\perp} \subseteq {C}$ if and only if $\mathcal{A}_i^{\perp} \subseteq \mathcal{A}_i$ for all $i=0,1,2$.
	\end{corollary}
\begin{theorem}\label{Thm quant}
	Let $C$ be a skew $(\sigma,\delta)$-constacyclic code of length $n$ over $\mathcal{R}$ with Gray image $\psi(C)$ which has parameters $[3n,k,d_{G}]$ where $d_{G}$ is the minimum Gray distance of $C$. If $C^{\perp}\subseteq C$, then there exists a quantum code with parameters $[[3n,2k-3n,d_{G}]]$ over $\mathbb{F}_{p^m}$.
\end{theorem}
\subsection{Computational results}
In this subsection, we provide an example to validate our results. A quantum error-correcting code $Q$ with parameters $[[n,k,d]]_{p^m}$ satisfies quantum Singleton bound $2d +k \leq n+2$. In case of equality, $Q$ is called quantum maximum-distance-separable(MDS) code. All the computations are performed by using the Magma computation system \cite{Magma}.
\begin{example} Let $\mathbb{F}_{25}=\mathbb{F}_{5}(t)$ where $t^2=t+3$ and $\mathcal{R}=\mathbb{F}_{5^2}+v\mathbb{F}_{5^2}+v^2\mathbb{F}_{5^2},\ v^3=v $. Let $\Theta$ be the Frobenius automorphism over $\mathbb{F}_{5^2}$ and $\sigma$ be the extension of $\Theta$ over $\mathcal{R}$, defined by $\sigma(a_0+a_1v+a_2 v^2)=\Theta(a_0)+\Theta(a_1)v+\Theta(a_2)v^2.$
	Let $\delta= 1-2v^2$ and $n=12$. Then $\lambda_0=1, \lambda_1=-1$ and $ \lambda_2=-1 $. Clearly, $\Theta$ fixes $\lambda_0, \lambda_1, \lambda_2$ and order of $\Theta$ divides $n$. Now, in $\mathbb{F}_{5^2}[x;\Theta]$, we have,
	\begin{align*}
	x^{12}-1=&(x^2 + (3t + 2)x + 2t + 1)(x^2 + 3t + 3)(x^2 + 2t + 1)\\
	& (x + t + 1)(x + 2t + 1)(x + 2t + 2)(x + t + 3)\\
	&(x^2 + (2t + 3)x + 3t + 3)\\
	x^{12}+1=&(x^2 + t + 1)(x^2 + 4t + 2)(x^2 + 4t + 4)(x^2 + t + 3)\\
	&(x + 4t)(x + 4t + 1)(x+ 4t + 3)^2.\\
	x^{12}+1=&(x^2 + t + 1)(x^2 + 4t + 2)(x^2 + 4t + 4)(x^2 + t + 3)\\
	&(x + 3t)(x + t + 4)(x + t)(x+3t+2)
	\end{align*}
	Let $f_0(x)=x^2 + (2t + 3)x + 3t + 3$, $f_1(x)=x + 4t + 3$ and $f_2(x)=x + 3t + 2$. Then ${C}=\langle \eta_0f_0(x)+\eta_1f_1(x) +\eta_2f_2(x) \rangle $ is a skew $(\sigma, \delta)$-constacyclic code over $\mathcal{R}$.
	 Here,
	\begin{align*}
	h_0(x)=&x^{10} + (3t + 2)x^9 + 2tx^8 + (2t + 3)x^7 + (3t + 4)x^6 \\
	+& x^4 + (3t + 2)x^3 + 2tx^2 + (2t + 3)x + 3t + 4,\\
	h_1(x)=&x^{11} + (4t + 3)x^{10} + 3x^9 + (2t + 4)x^8 + 4x^7\\
	 +& (t + 2)x^6 + 2x^5 + (3t + 1)x^4 + x^3 + (4t + 3)x^2\\
	 +& 3x + 2t + 4,\\
	h_2(x)=&x^{11} + 3tx^{10} + 3x^9 + 4tx^8 + 4x^7 + 2tx^6 + 2x^5 + tx^4 \\
	+& x^3 + 3tx^2 + 3x + 4t,\\
	h_0^*(x)=&(3t + 4)x^{10} + 3tx^9 + 2tx^8 + 2tx^7 + x^6 + (3t + 4)x^4\\
	 +& 3tx^3 + 2tx^2 + 2tx + 1,\\
	h_1^*(x)=&(3t + 1)x^{11} + 3x^{10} + (t + 2)x^9 + x^8 + (2t + 4)x^7\\
	 +& 2x^6 + (4t + 3)x^5 + 4x^4 + (3t + 1)x^3 + 3x^2\\
	 +& (t + 2)x + 1,\\
	h_2^*(x)= &(t + 4)x^{11} + 3x^{10} + (2t + 3)x^9 + x^8 + (4t + 1)x^7\\
	 +& 2x^6 + (3t + 2)x^5 + 4x^4 + (t + 4)x^3 + 3x^2\\
	  +& (2t + 3)x + 1,
	\end{align*}
	and
		\begin{align*}
	h^{*}_0(x)h_0(x)=&((3t + 4)x^8 + (2t + 1)x^6 + (3t + 4)x^2\\
	 +& 2t + 1)(x^{12}-1),\\
	h^{*}_1(x)h_1(x)=&((3t + 1)x^{10} + 4x^9 + (2t + 4)x^8+ (2t + 4)x^6\\
	+& x^5 + (3t + 1)x^4 + (3t + 1)x^2 + 4x\\
	 +& 2t + 4)(x^{12}+1),\\
	h^{*}_2(x)h_2(x)=&((t + 4)x^{10} + (3t + 1)x^9 + 4tx^8 + (4t + 1)x^6\\
	+& (2t + 4)x^5 + tx^4 + (t + 4)x^2 + (3t + 1)x\\
	 +& 4t)(x^{12}+1).
	\end{align*}

	Let  \[
	M=
	\left[ {\begin{array}{ccc}
		3&2&1 \\
		3&4&3 \\
		4&3&2
		\end{array} } \right]\in GL_3(\mathbb{F}_{5^2}),
	\] satisfying $MM^t=4I_3$. Then the Gray image $\psi({C})$ has the parameter $[36,32,3]$. Since, $h_0^{*}(x)h_{0}(x)$ is right divisible by  $(x^{12}-1)$ and $h_1^{*}(x)h_{1}(x), h_{2}^{*}(x)h_{2}(x)$ are right divisible by $(x^{12}+1)$. Thus, by Lemma \ref{thm dual}, $\mathcal{A}_i^{\perp}\subseteq \mathcal{A}_i$ for $i=0,1,2$, which implies $C^{\perp}\subseteq C$. Therefore, by Theorem \ref{Thm quant}, there exists a quantum code with parameter $[[36,28,3]]_{25}$.
\end{example}

In Table 1, we obtain many new quantum error-correcting codes from skew $(\sigma,\delta)$-constacyclic codes over $\mathcal{R}$. We use first column for writing the length of skew constacyclic code $C$ over $\mathcal{R}$ while second column is used for unit $\delta$ in $\mathcal{R}$. In third column we have written units $\lambda_i$ of $\mathbb{F}_{p^m}$ corresponding to $\delta$ whereas column fourth, fifth and sixth are used for writing the generator polynomials of skew $(\Theta,\lambda_i)$-constacyclic code $ \mathcal{A}_i$ over $\mathbb{F}_{p^m}$ for $i=0,1,2$, respectively. The parameter of the Gray images of skew $(\sigma,\delta)$-constacyclic codes are written in column seventh and last column denotes the constructed quantum codes. The coefficients of the generator polynomials $f_0(x),f_1(x)$ and  $f_2(x)$ in column $4,5$ and $6$ are given in ascending order, e.g., the polynomial $x^4+5tx^2+(2t+4)x+12$ as $(12)(2t+4)(5t)01$.\\

\begin{remark} In order to get the Gray images of skew constacyclic codes, we use the matrices
\begin{equation*}
\begin{bmatrix}
t^2&2&t \\
2&t&t^2 \\
t&t^2&2
\end{bmatrix},
\begin{bmatrix}
1&6&3 \\
4&1&6 \\
6&3&6
\end{bmatrix} ,
\begin{bmatrix}
7&4&2 \\
9&7&4 \\
4&2&4
\end{bmatrix} \text{and}
\begin{bmatrix}
9&4&2 \\
11&9&4 \\
4&2&4
\end{bmatrix}
\end{equation*}
for $\mathbb{F}_{3^2}$, $\mathbb{F}_{7^2}$, $\mathbb{F}_{11^{2}}$ and $\mathbb{F}_{13^{2}}$, respectively.
\end{remark}

\begin{table*}
	\caption{New Quantum codes $[[n,k,d]]_{p^m}$ from skew $(\sigma,\delta)$-constacyclic codes over the ring $\mathbb{F}_{p^m}+v\mathbb{F}_{p^m}+v^2 \mathbb{F}_{p^m}$}
	
	\renewcommand{\arraystretch}{1.5}
	\begin{center}
		
		\begin{tabular}{|c|c|c|c|c|c|c|c|c|}
			%{p{1cm} p{.5cm} p{4cm} p{4cm} p{1.5cm} p{1.5cm}}
			
			%\multicolumn{6}{c}{Table 1: Linear codes as Gray images of $\lambda$-constacyclic codes} \\\\
			\hline
			$n$ & $\delta$ & $(\lambda_0,\lambda_1,\lambda_2)$ & $f_0(x)$ & $f_1(x)$   &$f_2(x)$ & $\psi(\mathcal{C})$ & $[[n,k,d]]_{p^m}$ \\
			\hline
			$12$ & $1$ & $(1,1,1)$ & $11$ & $(t)1$    &$(2t+1)1$& $[36,33,3]$ & $[[36,30,3]]_{9}$ \\
			$6$ & $-1$ & $(-1,-1,-1)$ & $1(4t+3)1$ & $21$    &$31$& $[18,14,4]$ & $[[18,10,4]]_{25}$ \\
			$10$ & $1$ & $(1,1,1)$ & $(3t+4)1$ &$1(t+3)1$& $(3t+3)1$   &  $[30,26,3]$ & $[[30,22,3]]_{25}$ \\
			$12$ & $1-2v^2$ & $(1,-1,-1)$&$(3t+3)(2t+3)1$ & $(4t+3)1$ & $(3t+2)1$   &  $[36,32,3]$ & $[[36,28,3]]_{25}$ \\
			$8$ & $1-2v^2$ & $(1,-1,-1)$ & $(t+3)1$ &$(5t+6)(2t+2)1$& $(t+3)(3t+6)1$   &  $[24,19,4]$ & $[[24,14,4]]_{49}$ \\
			$14$ & $2v^2-1$ & $(-1,1,1)$ & $(5t+4)1$ & $(t+3)1$ &$(2t+1)(3t)1$  &  $[42,38,3]$ & $[[42,34,3]]_{49}$ \\
			$14$ & $2v^2-1$ & $(-1,1,1)$ & $(5t+4)1$ & $(t+3)1$ &$(2t+1)1$  &  $[42,39,2]$ & $[[42,36,2]]_{49}$ \\
			$18$ & $1-2v^2$ & $(1,-1,-1)$ & $(3t+2)1$ & $(5t+2)1$ &$(6t)1$  &  $[54,51,3]$ & $[[54,48,3]]_{49}$ \\
			$10$ & $1-2v^2$ & $(1,-1,-1)$ & $(5t+9)(9t+2)1$ & $(7t+4)1$  &$(3t+8)1$ &  $[30,26,3]$ & $[[30,22,3]]_{121}$ \\
			$16$ & $1$ & $(1,1,1)$ & $(4t+3)1$&$(10t+1)(9t+7)1$ & $(5t+8)(7t+4)1$   &  $[48,43,4]$ & $[[48,38,4]]_{121}$ \\
			$20$ & $2v^2-1$ & $(-1,1,1)$ & $(7t+10)(4t+9)1$ & $(8t+7)1$ &$(2t+7)(t)1$  &  $[60,55,3]$ & $[[60,50,3]]_{121}$ \\
			$4$ & $-1$ & $(-1,-1,-1)$ &$(12t+11)1$& $(9t+6)1$ & $(t+6)1$   &  $[12,9,4]$ & $[[12,6,4]]_{169}$ \\
			$6$ & $-1$ & $(-1,-1,-1)$ & $(t+11)1$ & $81$ &$(2t+8)1$   &  $[18,15,4]$ & $[[18,12,4]]_{169}$ \\
			$8$ & $2v^2-1$ & $(-1,1,1)$ & $(2t+11)(7t+3)1$ & $(t+10)1$ &$(5t+4)1$  &  $[24,20,3]$ & $[[24,16,3]]_{169}$ \\
			
			\hline
		\end{tabular}\label{tab1}
	\end{center}
\end{table*}
\vfill

\section{Conclusion}

In this paper, we have studied the structure of skew constacyclic codes and their duals over a finite commutative non-chain ring $\mathbb{F}_{p^m}+v\mathbb{F}_{p^m}+v^2 \mathbb{F}_{p^m}$ where $v^3=v$ and obtained many new non-binary quantum codes compared to the best-known codes. Therefore, the recent trend to obtain quantum codes from constacyclic codes can be enlarged to skew constacyclic codes successfully to get better codes. This leads to many open problems for the application of skew constacyclic codes to quantum codes over different finite non-chain rings that appeared in \cite{Ashraf18,Gao15,Gao19a,Islam18b,Li18,Ma19}.

\section*{Acknowledgment}
The authors are thankful to the Council of Scientific \& Industrial Research (CSIR), Govt. of India for financial supports and Indian Institute of Technology Patna, India for providing research facilities.

\end{document}